\documentclass[12pt,a4paper]{article} %

\newcommand{\lv}[1]{#1}
\newcommand{\sv}[1]{}

\lv{
  \usepackage[T1]{fontenc}
}

\usepackage{graphicx}
\usepackage{hyperref}

\usepackage{amsmath}
\usepackage{amssymb}
\usepackage{xspace}
\usepackage{enumerate}
\usepackage{complexity}
\lv{
  \usepackage{amsthm}
}
\usepackage[capitalize,noabbrev]{cleveref}

\usepackage{todonotes}
\usepackage{mathtools}

\sv{
 \usepackage{apxproof}
\ifx\proof\inlineproof
  \let\apxmark\relax  
\else
  \def\apxmark{{\sf\bfseries$\!$*\,}}
\fi
\newtheoremrep{theorem2}[theorem]{Theorem}
\newtheoremrep{proposition2}[theorem]{Proposition}
\newtheoremrep{corollary2}[theorem]{Corollary}
\newtheoremrep{lemma2}[theorem]{Lemma}        
\newif\ifblindre
 \blindrefalse 
 }
 \lv{
  \let\apxmark\relax  
\newtheorem{definition}{Definition}[section]
\newtheorem{theorem}[definition]{Theorem}
\newtheorem{lemma}[definition]{Lemma}
\newtheorem{proposition}[definition]{Proposition}
\newtheorem{claim}[definition]{Claim}
\newtheorem{remark}[definition]{Remark}

 }

\newcommand{\Ucr}{\ensuremath{\mathrm{ucr}}}
\newcommand{\Ucrk}[1]{\ensuremath{\mathrm{ucr}_{#1}}}
\newcommand{\U}{\ensuremath{\mathrm{unc}}}
\newcommand{\Um}{\ensuremath{\mathrm{ounc}}}

\def\ca#1{\mathcal{#1}}

\def\crd{\operatorname{cr}} %
\def\crg{\operatorname{cr}} %
\def\tcrn{\mathop{\rm tcr}}

\lv{
  \usepackage{authblk}
  \usepackage[a4paper,left=3.5cm,right=3.5cm, top=3cm, bottom=4cm]{geometry}
}

\tikzstyle{vertex}=[circle, fill, inner sep=1pt]
\usetikzlibrary{patterns,patterns.meta,shapes,backgrounds,positioning}
\usetikzlibrary{math} %

\begin{document}
\title{Minimizing an Uncrossed Collection of Drawings%
  \thanks{An extended abstract of this manuscript appeared in the Proceedings of the 31st International Symposium on Graph Drawing and Network Visualization (GD 2023)~\cite{OurGD}.}
}

\lv{
  \author[1]{Petr Hliněný}%
  \author[2]{Tomáš Masařík%
  \thanks{Tomáš Masařík was supported by the Polish National Science Centre SONATA-17 grant number 2021/43/D/ST6/03312.}}
  \affil[1]{Faculty of Informatics of Masaryk University, Brno, Czech Republic 

  \texttt{hlineny@fi.muni.cz}}
  \affil[2]{Institute of Informatics, Faculty of Mathematics, Informatics and Mechanics, University~of~Warsaw, Warszawa, Poland 

  \texttt{masarik@mimuw.edu.pl}}
  \date{}
}
\maketitle              %
\begin{abstract}
In this paper, we introduce the following new concept in graph drawing.
Our task is to find a small collection of drawings such that they all together satisfy some property that is useful for graph visualization.
We propose investigating a property where each edge is not crossed in at least one drawing in the collection.
We call such collection \emph{uncrossed}.
This property is motivated by a quintessential problem of the crossing number, where one asks for a drawing where the number of edge crossings is minimum.
Indeed, if we are allowed to visualize only one drawing, then the one which minimizes the number of crossings is probably the neatest for the first orientation.
However, a collection of drawings where each highlights a different aspect of a graph without any crossings could shed even more light on the graph's structure.

We propose two definitions. 
First, the \emph{uncrossed number}, minimizes the number of graph drawings in a collection, satisfying the uncrossed property.
Second, the \emph{uncrossed crossing number}, minimizes the total number of crossings in the collection that satisfy the uncrossed property.
For both definitions, we establish initial results.
We prove that the uncrossed crossing number is NP-hard, but there is an \FPT{} algorithm parameterized by the solution size.

\sv{%
}
\lv{%
  \smallskip{} \begin{center} 
\end{center}
\smallskip
}
\end{abstract}

\section{Introduction}

Determining the \emph{crossing number} $\crg(G)$ of a graph $G$, i.e.\ the smallest possible number $k=\crg(G)$ of pairwise transverse intersections 
(called {\em crossings}) of edges in any drawing of~$G$ is among the most important theoretical problems in graph drawing.
Its computational complexity is well-researched.
It is known that graphs with crossing number $k=0$, i.e.\ planar graphs, can be easily recognized.
Computing the crossing number $k$ of a graph is \NP-hard~\cite{GareyJ83}, 
even in very restricted settings of cubic graphs~\cite{Hlineny06}, of a fixed rotation system~\cite{DBLP:journals/algorithmica/PelsmajerSS11}
and of near-planar graphs~\cite{CabelloM13},
It is also \APX-hard~\cite{Cabello13}.
However, there is a fixed-parameter algorithm for the problem~\cite{Grohe04} with respect to the number of crossings~$k$, 
and even one that can solve the problem in linear time $\mathcal{O}(f(k)n)$~\cite{KawarabayashiR07}.

More recently, various {graph drawing extension} problems motivated by graph and network visualization received increased attention.
In relation to the traditional crossing number problem, we in particular, mention works on the edge insertion problem
\cite{DBLP:conf/icalp/ChimaniH11,DBLP:conf/compgeom/ChimaniH16,DBLP:journals/algorithmica/GutwengerMW05}
and on the crossing number in a partially predrawn setting \cite{GanianHKPV21,DBLP:conf/compgeom/HammH22}.
Such drawing extension solutions can be useful, for instance, in situations in which a particular region or a feature of a graph should be highlighted,
and the rest of the drawing is then completed as nicely as possible while respecting the highlighted part.
It is, however, more common that we want to somehow nicely display every part of the given graph.
That is, of course, problematic with only one drawing, but what if we are allowed to produce several drawings of the same graph or network at once?
On the experimental side, the same idea was actually raised by Biedl, Marks, Ryall, and Whitesides~\cite{DBLP:conf/gd/BiedlMRW98} quite long time ago, 
but does not seem to be actively pursued nowadays.

Motivated by the latter thoughts, this paper proposes and studies the following concept of looking at graph drawings. 
Instead of finding a single ``flawless'' drawing (in the plane), we look for a small collection of drawings that fulfill some visualization property(ies) ``together'', meaning in their union. 
In that way, each drawing in the collection spotlights a different part of the graph's structure.
This seems to be a promising concept as minimizing just one property of a drawing globally, which is traditionally done, 
might produce a drawing that is nowhere nice locally.
Therefore, having to produce only one solution may lead to visualizations that are not as easy to understand.

In the proposed approach, we aim to produce a (smallest) collection of drawings of a graph that all together admit a property useful for visualizing the graph structure.
There might be several different sensible candidates for the property.
Inspired by the crossing number, we consider a property that requires each edge not to be crossed in at least one drawing in the collection.
Therefore, we produce a collection of drawings that emphasize different aspects of a graph in different drawings.
Formally, we define the uncrossed number of a graph as follows (see \Cref{sec:prelim} for the necessary standard definitions).

\begin{definition}[Uncrossed number]
  Let $G=(V,E)$ be a graph.
  A family of drawings $D_1,\ldots,D_{k}$ of~$G$,
  such that each edge $e\in E$ is not crossed in some drawing $D_i$, $1\leq i\leq k$, is called an \emph{uncrossed collection of drawings} of~$G$.
  The \emph{uncrossed number} of a graph $G$, denoted by $\U(G)$, is defined as the least cardinality $k$ of an uncrossed collection of drawings of~$G$.
\end{definition}

The uncrossed number is related to the traditional \emph{thickness} of a graph (the thickness of $G$ ($\theta(G)$) is the least number of planar subgraphs of $G$ whose union is whole~$G$).
Obviously, the thickness of $G$ is at most $\U(G)$, and $G$ is of thickness~$1$ (planar) if and only if $\U(G)=1$.
On the other hand, the uncrossed number is upper bounded by \emph{outerthickness} of a graph (the outerthickness of $G$ ($\theta_o(G)$) is the least number of outerplanar subgraphs of $G$ whose union is whole~$G$).
Indeed, we can draw each of the outerplanar graphs without crossings and draw all the other edges of $G$ inside the outerface.
It is valuable to know that a planar graph can be edge partitioned into two outerplanar graphs as proven by Gonçalves in 2005~\cite{Goncalves}.
We summarize all the relations above; For any $G$: \[\theta(G)\le \U(G)\le \theta_o(G)\le2\theta(G).\]
To give a bit more intuition, we argue in Proposition~\ref{prop:k7} that $\U(K_7)=3$ which is the same bound as in the outerthickness, $\theta_o(K_7)=3$~\cite{Guy1990}, while indeed $\theta(K_7)=2$.

Thickness is a very well-studied graph parameter; consult an initial survey~\cite{DBLP:journals/gc/MutzelOS98} or more recent papers~\cite{DjumovicWood,Duncan} for further details.
It used to be on Garey \& Johnson's list of open problems but was proven to be \NP-complete already in 1983 by Mansfield~\cite{Mansfield}. Surprisingly, the complexity of outerthickness has yet to be uncovered.
A recent paper by Xu and Zha~\cite{XuZha} provides bounds based on the genus of the graph, but in many other aspects, outerthickness still needs to be explored.

Analogously as the uncrossed number is an alternative counterpart of planarity in our approach,
the graph crossing number has a natural counterpart in the total crossing number of an uncrossed collection of drawings.
Informally, the task is to find a collection of drawings satisfying the uncrossed property such that it has the smallest number of crossings summed over the collection.

\begin{definition}[Uncrossed crossing number]\label{def:uncrcr}
  Let $G=(V,E)$ be a graph and $k\geq1$ an integer.
  The \emph{uncrossed crossing number in $k$ drawings} of a graph~$G$, denoted by $\Ucrk{k}(G)$, 
  equals the least sum $\crd(D_1)+\cdots+\crd(D_{k'})$ over all uncrossed collections $D_1,\ldots,D_{k'}$ of drawings of~$G$ where~$k'\leq k$,
  or \mbox{$\Ucrk{k}(G)=\infty$} if no uncrossed collection exists.

  The \emph{uncrossed crossing number} of $G$ is defined as $\Ucr(G)\coloneqq \min_k\Ucrk{k}(G)$.
  The least $k$ for which $\Ucr(G)=\Ucrk{k}(G)$ is called the \emph{crossing-optimal uncrossed number} $\Um(G)$.
\end{definition}

Observe that $\Ucr(G)$ is a well-defined integer since there always exists an uncrossed collection of $k=|E(G)|$ drawings of~$G$.
Furthermore, $\U(G)\le\Um(G)$,
and $\Um(G)\leq|E(G)|$ since every uncrossed collection of $\,>\!|E(G)|$ drawings obviously contains a strict uncrossed subcollection.
On the other hand, one can easily construct a family of graphs $G$ such that $\U(G)=2$ and $\Um(G)$ is unbounded; see Proposition~\ref{prop:unc2large} and \Cref{fig:cycleMob}.
We also have obvious $\Ucr(G)\geq\Um(G)\cdot\crg(G)$, but the gap in this inequality can be arbitrarily large---consider, e.g., the complete graph $K_5$ with every edge except two disjoint ones blown up to many parallel edges, as in Proposition~\ref{prop:ucr2large}. 

We also provide some useful estimates on the uncrossed crossing number.
We, for instance, give the asymptotics of the uncrossed crossing number for the complete graph in Proposition~\ref{prop:Kn},
and we give an asymptotically tight analogy of the classical Crossing Lemma for the uncrossed crossing number of graphs with sufficiently many edges in \Cref{thm:lbUcr}.

As discussed, the crossing number is \NP-hard even in various restricted settings.
We prove that the uncrossed crossing number is \NP-complete in \Cref{sec:hard} building on ideas of Hliněný and Derňár~\cite{DBLP:conf/compgeom/HlinenyD16}.
We complement this result by presenting an \FPT{} algorithm parameterized by its size in \Cref{sec:FPT}.
This is analogous to Grohe's \FPT{} algorithm for the crossing number parameterized by its size~\cite{Grohe04}.

\section{Preliminaries}\label{sec:prelim}

In this paper, we consider multigraphs by default, i.e., our graphs are allowed to have multiple edges (while loops are irrelevant here),
with the understanding that we can always subdivide parallel edges without changing the crossing number.

A \emph{drawing} $D$ of a graph $G$ in the Euclidean plane $\mathbb{R}^2$ is a function that maps each vertex $v \in V(G)$ to a distinct point 
$D(v) \in \mathbb R^2$ and each edge $e=uv \in E(G)$ to a simple open curve $D(e) \subset \mathbb R^2$ with the ends $D(u)$ and $D(v)$.
We require that $D(e)$ is disjoint from $D(w)$ for all~$w\in V(G)\setminus\{u,v\}$.
Throughout the paper, we will moreover assume that:
there are finitely many points that are in an intersection of two edges, no more than two edges intersect in any single point other than a vertex,
and whenever two edges intersect in a point, they do so transversally (i.e., not tangentially).

An intersection (a common point) of two edges is called a \emph{crossing} of these edges.
A drawing $D$ is \emph{planar} (or a {\em plane graph}) if $D$ has no crossings, and a graph is \emph{planar} if it has a planar drawing.
The number of crossings in a drawing $\ca G$ is denoted by $\crd(D)$.
The (ordinary) {\em crossing number $\crg(G)$ of $G$} is defined as the minimum of $\crd(D)$ over all drawings $D$ of~$G$.

The following is a useful artifice in crossing numbers research. 
In a {\em weighted} graph, each edge is assigned a positive integer (the {\em weight} of the edge).
Now the {\em weighted crossing number} is defined as the ordinary crossing number, but a crossing between edges $e_1$ and $e_2$,
say of weights $t_1$ and $t_2$, contributes the product $t_1\cdot t_2$ to the weighted crossing number.
For the purpose of computing the crossing number, an edge of integer weight $t$ can be equivalently replaced by a bunch of $t$ parallel edges of weights~$1$;
this is since we can easily redraw every edge of the bunch tightly along the ``cheapest'' edge of the bunch.
The same argument holds also with our new \Cref{def:uncrcr}.

\section{Basic properties}

To give the readers a better feeling about the new concept(s) before moving onto the computational-complexity properties, 
we first present some basic properties of the uncrossed and the uncrossed crossing number.

We start by inspecting the uncrossed number.

\begin{proposition}\label{prop:k7}
  We have $\U(K_7)=3$.
\end{proposition}

\begin{proof}
  We say that drawing \emph{realizes} an edge if it is not crossed in that drawing.
  Suppose that drawings $D_1$ and $D_2$ exist.
Let $D_1$ be the drawing that realizes at least as many edges as $D_2$.
As a graph on the uncrossed edges needs to be planar, the number of realized edges in $D_1$ ranges from 11 (As $K_7$ has $21$ edges in total) to 15 ($3\cdot 7 -6$).
Consider faces that are defined by edges that are not crossed in drawing $D_i$, $i\in\{1,2\}$.
We call them \emph{uncrossed faces} of $D_i$ and we denote by $D_i'$ the subdrawing of $D_i$ consisting of all uncrossed faces.

 As we draw $K_7$, for each pair of vertices, there needs to be a common uncrossed face containing both.
 We call it the \emph{neighborhood property}.

 \begin{claim}%
   Drawing $D'_1$ consists of the outerface with at most one vertex inside.
   We denote such vertex \emph{central} (in drawing $D_1'$).
   Moreover, the central vertex (if it exists) is contained in all other faces except for the outerface.
\end{claim}
\noindent\textit{Proof of the claim.}
Observe that each vertex needs to have at least two realized edges.
If it has at most one then all other vertices have to be on one common face and hence they form an outerplanar graph, so there are at most $2\cdot6-3+1<11$ realized edges.
Assume that the outerface is a longest face and observe that there are no chords as this violates the neighborhood property.

Suppose that there are two vertices inside the outerface.
  Then those need to realize at least six edges on their neighborhood which is not possible as each vertex inside can share uncrossed edges with at most two vertices of the outerface otherwise we violate the neighbourhood property.

  Suppose that there are three vertices inside the outerface.
  Then those need to realize at least seven edges on their neighborhood.
  If there is a vertex $v$ that does not realize any edge with the outerface then $v$ is on a face of length at least $5$ or violates the neighborhood property.
  Now, let $v$ be a vertex that realizes exactly $1$ edge with the outerface.
  To preserve the neighborhood property and not to create a face longer than $4$, $v$ needs to have $3$ realized edges.
  Let $v_1,v_2$ be the other two vertices inside the outerface.
  If edge $v_1v_2$ is realized then both $v_1,v_2$ cannot realize more than $1$ edge on the outerface which is not enough.
  Hence, $v_1v_2$ is not realized, but then $v_1$ and $v_2$ have to both realize an edge with the same vertex on the outerface, else they create a face longer than $4$.
  However, then there is no common face of either $v_1$ or $v_2$ with a vertex on the outerface.

  If there are four vertices inside the outerface then the uncrossed edges need to form a triangulation which is not possible as this means no crossing edges in the drawing.
\hfill$\diamondsuit$

\lv{\medskip}

  The claim leaves only three possibilities for $D_1'$; 
  \begin{description}
    \item[\rm Case 1:]\phantomsection\label{c1} $D'_1$ form an outerplanar triangulation drawing (case of 11 realized edges), but due to $\theta_o(K_7)=3$~\cite{Guy1990} $D'_2$ could not be outerplanar in this case,
    \item[\rm Case 2:]\phantomsection\label{c2} $D'_1$ has one central vertex contained in 4 uncrossed triangles, one uncrossed 4-cycle, and the outerface has length 6 (case of 11 realized edges),
    \item[\rm Case 3:]\phantomsection\label{c3} $D'_1$ has one central vertex contained in 6 uncrossed triangles, and the outerface has length 6 (case of 12 realized edges).
  \end{description}

  In Cases \hyperref[c2]{2} and \hyperref[c3]{3}, we need to realize at least three edges of each vertex except for the central one $v_c$ (of $D'_1$).
Moreover, in $D'_2$, we need to put $v_c$ to one face, even though we do not need to realize more than one of its edges, say that this will be the outerface.
As all the other vertices are on a common face, their realized edges form an outerplanar graph on six vertices.
Therefore, there are two vertices of degree $2$ in $D'_2$ which contradicts the option that for each of them, we need to realize three edges.

In \hyperref[c1]{Case 1}, we already know that $D'_2$ cannot be outerplanar.
As $D'_1$ is outerplanar there has to be two non-adjacent vertices $v_1$ and $v_2$ that each realize $4$ edges in $D_2$. This implies that there are at least three vertices as their common neighbors in $D'_2$ and thus they create three separate faces.
No matter which face we put the remaining vertices in they will be separated from the vertex not on that face.
\end{proof}

\smallskip
We continue with a few easy facts about the uncrossed crossing number.

\begin{figure}[ht]
$$
\begin{tikzpicture}[xscale=1,yscale=1]
\tikzstyle{every node}=[draw, shape=circle, minimum size=2pt,inner sep=1.4pt, fill=black]
\node at (0,2.2) (x1) {};
\node at (0.9,2.1) (x2) {}; \node at (1.7,1.8) (x3) {};
\node at (2.4,1.3) (x4) {}; \node at (2.85,0.65) (x5) {};
\node at (3.0,0) (x6) {};
\node at (2.4,-1.3) (x7) {}; \node at (2.85,-0.65) (x8) {};
\node at (0.9,-2.1) (x9) {}; \node at (1.7,-1.8) (x10) {};
\node at (0,-2.2) (y1) {};
\node at (-0.9,-2.1) (y2) {}; \node at (-1.7,-1.8) (y3) {};
\node at (-2.4,-1.3) (y4) {}; \node at (-2.85,-0.65) (y5) {};
\node at (-3.0,0) (y6) {};
\node at (-2.4,1.3) (y7) {}; \node at (-2.85,0.65) (y8) {};
\node at (-0.9,2.1) (y9) {}; \node at (-1.7,1.8) (y10) {};
\tikzstyle{every path}=[draw, fill=none]
\foreach \i in {1,...,10} { \draw (x\i) to[bend right=6] (y\i) ; }
\begin{scope}[on background layer]
\draw[color=gray, ultra thick,fill=none] (0,0) ellipse (30mm and 22mm);
\end{scope}[on background layer]
\end{tikzpicture}
\quad
\begin{tikzpicture}[xscale=0.9,yscale=0.8]
\tikzstyle{every node}=[draw, shape=circle, minimum size=2pt,inner sep=1.4pt, fill=black]
\node at (3,1.5) (x1) {};
\node at (2.4,1.2) (x2) {}; \node at (1.8,0.9) (x3) {};
\node at (1.2,0.6) (x4) {}; \node at (0.6,0.3) (x5) {};
\node at (-0.6,-0.3) (x6) {};
\node at (-1.2,-0.6) (x7) {}; \node at (-1.8,-0.9) (x8) {};
\node at (-2.4,-1.2) (x9) {}; \node at (-3,-1.5) (x10) {};
\node at (3,-1.5) (y1) {};
\node at (2.4,-1.2) (y2) {}; \node at (1.8,-0.9) (y3) {};
\node at (1.2,-0.6) (y4) {}; \node at (0.6,-0.3) (y5) {};
\node at (-0.6,0.3) (y6) {};
\node at (-1.2,0.6) (y7) {}; \node at (-1.8,0.9) (y8) {};
\node at (-2.4,1.2) (y9) {}; \node at (-3,1.5) (y10) {};
\tikzstyle{every path}=[draw, fill=none]
\foreach \i in {1,...,10} { \draw (x\i) to[bend left=9] (y\i) ; }
\begin{scope}[on background layer]
\draw[color=gray, ultra thick,fill=none] (-3,1.5) to[out=150,in=30] (3,1.5) -- (-3,-1.5)
	to[out=-150,in=-30] (3,-1.5) -- (-3,1.5);
\end{scope}[on background layer]
\end{tikzpicture}
$$
\caption{A illustration of the graph $G$, for $m=10$, from the proof of Proposition~\ref{prop:unc2large};
	this graph has $\U(G)=2$ based on an uncrossed collection of two drawings which both look as depicted on the right.}
\label{fig:cycleMob}
\end{figure}
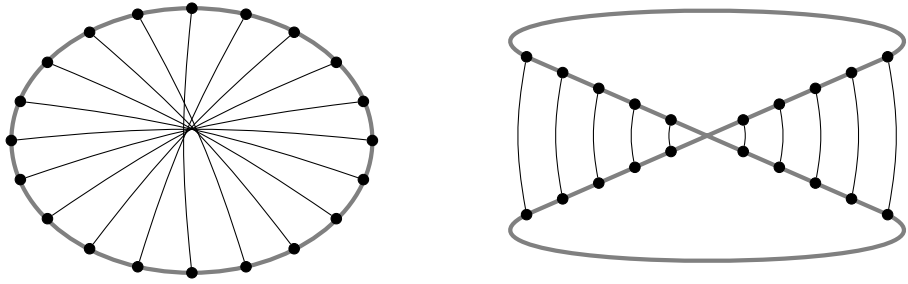

\sv{%
\begin{proposition2rep}\label{prop:unc2large}%
  For every integer $m$ there exists a graph $G$ on $2m$ vertices such that ${\U(G)=2}$ and $\Um(G)\geq m$ (see \Cref{fig:cycleMob}).
\end{proposition2rep}
}
\lv{%
\begin{proposition}\label{prop:unc2large}%
  For every integer $m$ there exists a graph $G$ on $2m$ vertices such that ${\U(G)=2}$ and $\Um(G)\geq m$ (see \Cref{fig:cycleMob}).
\end{proposition}
}
\begin{proof}
Let $G$ be constructed as follows: start with the cycle $C=C_{2m}$ on $2m$ vertices, make every edge of this cycle ``heavy'' of weight $m^3$,
and add the $m$ edges joining the opposite pairs of vertices on the cycle.
As \Cref{fig:cycleMob} easily shows, we have $\U(G)=2$.
There also exists an uncrossed collection of $m$ drawings of~$G$, each with $m-1\choose2$ crossings, obtained by separately ``flipping''
each one of the chords of the cycle $C$ in \Cref{fig:cycleMob} (left) to the outer face.
So, $\Ucr(G)\leq m{m-1\choose2}$.

Since $2\cdot\Ucr(G)<m^3$, any collection of drawings attaining $\Ucr(G)$ must leave the cycle $C$ (of weight at least~$1$) uncrossed.
And since $C$ with any three of the $m$ chords form a subdivision of nonplanar $K_{3,3}$, at most one of the chords can stay uncrossed
in a drawing of $G$ with uncrossed~$C$. Hence, $\Um(G)\geq m$.
\end{proof}

\sv{%
\begin{proposition2rep}\label{prop:ucr2large}%
For every integer $m\geq3$ there exists a graph $G$ such that \mbox{$\crg(G)=1$} and $\Ucr(G)=2m$.
Namely, $G$ can be obtained from the graph $K_5$ by selecting two disjoint edges and making every other edge ``heavy'' of weight~$m$.
\end{proposition2rep}
}
\lv{%
\begin{proposition}\label{prop:ucr2large}%
For every integer $m\geq3$ there exists a graph $G$ such that \mbox{$\crg(G)=1$} and $\Ucr(G)=2m$.
Namely, $G$ can be obtained from the graph $K_5$ by selecting two disjoint edges and making every other edge ``heavy'' of weight~$m$.
\end{proposition}
}
\begin{proof}
Let the two selected disjoint edges of $K_5$ in $G$ be $e_1$ and~$e_2$.
Then $\crg(G)=1$ witnessed by a drawing of $G$ in which only $e_1$ crosses with $e_2$,
and $\Ucr(G)\leq2m$ witnessed by a pair of drawings, one with $e_1$ crossing a heavy edge and other with $e_2$ crossing another heavy edge.
On the other hand, every uncrossed collection of drawings of $G$ must contain a drawing $D_1$ with uncrossed~$e_1$ and a drawing $D_2$ with uncrossed~$e_2$.
It cannot be $D_1=D_2$ since then $D_1$ would have a crossing of two other edges of weight $m\cdot m>2m$.
So each of $D_i$, $i=1,2$, has a crossing of $e_i$ with some $f_i$ of weight $m$;
in detail, if any of the $m$ edges parallel edges representing $f_i$ was free of crossings with other edges of $G$, then
we could redraw the remaining parallel edges representing $f_i$ closely along it to avoid such otherwise unavoidable crossing.
Consequently, we get $\Ucr(G)\geq\crg(D_1)+\crg(D_2)\geq2m$, as desired, and also~$\Um(G)=2$.
\end{proof}

With respect to the next claim, we remark that the exact value of $\crg(K_n)$ for every $n$ is still an open problem, 
although the right asymptotic $\crg(K_n)\in\Theta(n^4)$ follows directly from the famous Crossing Lemma \cite{ajtaiChvatalNewbornSzemeredi82,leighton83}.

\lv{
\begin{proposition}\label{prop:Kn}%
For every $n\geq5$, we have $\frac n6\cdot\crg(K_n)\leq \Ucr(K_n) \leq \frac{n^5}{96}$.
In particular, $\Ucr(K_n)\in\Theta(n^5)$.
\end{proposition}
}
\sv{%
\begin{proposition2rep}\label{prop:Kn}%
For every $n\geq5$, we have $\frac n6\cdot\crg(K_n)\leq \Ucr(K_n) \leq \frac{n^5}{96}$.
In particular, $\Ucr(K_n)\in\Theta(n^5)$.
\end{proposition2rep}
}
\begin{proof}

If we have an uncrossed collection of drawings of $K_n$, then every drawing has at least $\crg(K_n)$ crossings and at most $3n-6$ uncrossed edges.
Hence $$\Ucr(K_n)\geq \frac1{3n-6}{n\choose2}\cdot\crg(K_n)\geq \frac n6\cdot\crg(K_n).$$
On the other hand, denoting the vertices of $K_n$ by $x_1,\ldots,x_n$, we may choose a special spanning path $P_1:=(x_1,x_2,x_n,x_3,x_{n-1},\ldots,x_{\lceil n/2\rceil+1})$.
If we ``rotate'' $P_1$ $\lceil n/2\rceil$-times inside $K_n$ -- meaning to cyclically shift the indices
of~$x_{\star}$ in the path by $0,1,\ldots,$ $\lceil n/2\rceil-1$ modulo~$n$, we altogether cover all edges of $K_n$ at least once.
These copies of $P_1$ form a basis for our uncrossed collection of drawings;
up to automorphism, we draw $P_1$ uncrossed and horizontal, and distribute the remaining edges of $K_n$ above or below $P_1$ based on the parity of the left end.
One such drawing has the following number of crossings:
\begin{align*}
\sum_{i=1}^{n}\sum_{j=i+1}^{n} \left\lceil\frac{j-i-1}{2}\right\rceil(n-j) ~\leq&~
	\sum_{i=1}^{n}\sum_{j=i+1}^{n} \frac12(j-i-1)(n-j)
\\=&~ \frac1{48}n^4-\frac18n^3+\frac{11}{48}n^2-\frac18n
\end{align*}
Altogether, our collection has this upper bound on the number of crossings:
$$\Ucr(K_n) \leq \left(\frac1{48}n^4-\frac18n^3+\frac{11}{48}n^2-\frac18n\right)\cdot\left(\frac n2 +\frac12\right) \leq \frac1{96}n^5
.$$
(Assuming the natural conjectured value of $\crg(K_n)$, this gives the asymptotic ratio of $\frac14$ between the lower and upper bounds.)
\end{proof}

The simple lower-bound proof from Proposition~\ref{prop:Kn} can be easily generalized with the help of the aforementioned Crossing Lemma, and the
generalization is asymptotically tight by the example of~$K_n$ in Proposition~\ref{prop:Kn}.
\begin{theorem}\label{thm:lbUcr}
Let $G$ be a simple graph such that $|E(G)|\geq7|V(G)|$.
Then $\Ucr(G)\geq|E(G)|^4/\big(87\cdot|V(G)|^3\big)$.
\end{theorem}

\begin{proof}
Let $D_1,\ldots,D_k$ be any uncrossed collection of drawings of $G$.
By the currently best variant of the Crossing Lemma -- in Ackerman~\cite{DBLP:journals/comgeo/Ackerman19},
we for all $i=1,\ldots,k$ have $\crd(D_i)\geq|E(G)|^3/\big(29\cdot|V(G)|^2\big)$.
Moreover, since a simple graph $G$ can have at most $3|V(G)|-6$ uncrossed edges in any plane drawing,
we get $k\geq|E(G)|/\big(3|V(G)|-6\big)\geq|E(G)|/\big(3|V(G)|\big)$, which concludes the universal lower bound
$\crd(D_1)+\ldots+\crd(D_k)\geq|E(G)|^4/\big(87\cdot|V(G)|^3\big)$.
\end{proof}

\section{Hardness reduction}\label{sec:hard}

Our hardness proof for the uncrossed crossing number builds on the ideas of \cite{DBLP:conf/compgeom/HlinenyD16}
(\Cref{thm:diaghard1}), which in turn is based on~\cite{CabelloM13}.

In the context of~\cite{DBLP:journals/jgt/PinontoanR03} we define a {\em tile}  $T=(H,a,b,c,d)$ where $H$ is a graph and $a,b,c,d$ 
is a sequence of distinct vertices, called here the {\em corners} of the tile.
The {\em inverted} tile $T^{\updownarrow}$ is the tile $(H,a,b,d,c)$.
A tile $T$ is {\em perfectly connected} if both $H$ and the subgraph $H-\{a,b,c,d\}$ are connected and no edge of $H$ has both ends in $\{a,b,c,d\}$.
A {\em tile drawing} of a tile $T=(H,a,b,c,d)$ is a drawing of the underlying graph $H$ in the
unit square such that the vertices $a,b,c,d$ are the upper left, lower left, lower right, and upper right corner, respectively.
The {\em tile crossing number $\tcrn(T)$} of a tile $T$ is the minimum number of crossings over all tile drawings of $T$.
A tile $T$ is {\em planar} if $\tcrn(T)=0$.

\begin{theorem}[{\cite[Definition~9 and Corollary~12]{DBLP:conf/compgeom/HlinenyD16}}]\label{thm:diaghard1}
The following problem is \NP-hard: given an integer $k$, and a perfectly connected planar tile~$T=(H,a,b,c,d)$ 
such that there exist $2k+1$ edge-disjoint paths between the vertices $a$ and $c$ in $H$ and the vertex $d$ is of degree~$1$, 
decide whether $\tcrn(T{}^{\updownarrow})\leq k$.
\end{theorem}

\begin{figure}[t]
$$
\begin{tikzpicture}[xscale=1,yscale=1.09]
\tikzstyle{every node}=[draw, shape=circle, minimum size=2pt,inner sep=1.4pt, fill=black]
\tikzstyle{every path}=[fill=lightgray]
\draw (2.5,2) -- (5,2) to [out=200,in=110] (5,0) -- (2.7,0) to [out=-5,in=290] (2.5,2) ;
\draw (0,2) to [out=290,in=50] (0,0) -- (2.5,2);
\draw (2.5,2) to [out=240,in=145] (2.33,0.15) -- (0,2);
\tikzstyle{every path}=[color=black]
\draw (2.5,2) node[label=above:{$\!\!c=c'\!\!$}, fill=red] (c) {};
\draw (2.5,0) node[label=below:{$\!\!\!d=d'\!\!\!$}] (d) {};
\draw (0,0) node[label=left:$b$] (b) {};
\draw (0,2) node[fill=red, label=left:$a$] (a) {};
\draw (5,0) node[fill=red, label=right:$a'$] (aa) {};
\draw (5,2) node[label=right:$b'$] (bb) {};
\draw (b)--(c) (2.33,0.15) -- (d) -- (2.7,0);
\draw[ultra thick, color=gray] (a)--(c) (c)--(aa);
\draw (1.25,1.6) node[draw=none, fill=none] {$H$} ;
\draw (3.75,1.6) node[draw=none, fill=none] {$H'$} ;
\draw (1.25,-0.25) node[draw=none, fill=none] {$(T{}^{\updownarrow})$} ;
\draw (3.75,-0.25) node[draw=none, fill=none] {$(T')$} ;
\begin{scope}[on background layer]
\draw[color=white!80!black, ultra thick] (2.5,1) ellipse (30mm and 18mm);
\end{scope}[on background layer]
\end{tikzpicture}
\qquad
\begin{tikzpicture}[xscale=1,yscale=1.09]
\tikzstyle{every node}=[draw, shape=circle, minimum size=2pt,inner sep=1.4pt, fill=black]
\tikzstyle{every path}=[fill=lightgray]
\draw (0,0) -- (2.5,0) to [out=110,in=175] (2.3,2) -- (0,2) to [out=-70,in=20] (0,0) ;
\draw (2.5,0) to [out=70,in=325] (2.68,1.85) -- (5,0);
\draw (5,0) to [out=110,in=230] (5,2) -- (2.5,0);
\tikzstyle{every path}=[color=black]
\draw (2.5,2) node[label=above:{$\!\!d=d'\!\!$}] (d) {};
\draw (2.5,0) node[label=below:{$\!\!\!c=c'\!\!\!$}, fill=red] (c) {};
\draw (0,0) node[label=left:$b$] (b) {};
\draw (0,2) node[fill=red, label=left:$a$] (a) {};
\draw (5,0) node[fill=red, label=right:$a'$] (aa) {};
\draw (5,2) node[label=right:$b'$] (bb) {};
\draw (aa) -- (d) -- (2.3,2);
\draw[ultra thick, color=gray] (a)--(c) (c)--(aa);
\draw (1.25,0.5) node[draw=none, fill=none] {$H$} ;
\draw (3.75,0.5) node[draw=none, fill=none] {$H'$} ;
\draw (1.25,-0.25) node[draw=none, fill=none] {$(T)$} ;
\draw (3.75,-0.25) node[draw=none, fill=none] {$(T'{}^{\updownarrow})$} ;
\begin{scope}[on background layer]
\draw[color=white!80!black, ultra thick] (2.5,1) ellipse (30mm and 18mm);
\end{scope}[on background layer]
\end{tikzpicture}
$$
\caption{A sketch of an uncrossed collection of two drawings of the graph $G$ in the proof of \Cref{thm:uncrohard2};
	the heavy outer cycle $C$ and the many $a$--$c$ and $a'$--$c'$ paths are outlined in thick gray.
	Only edges of $H$ are crossed on the left and only edges of $H'$ on the right.
If $\tcrn(T{}^{\updownarrow})\leq k$, then $\Ucr(G)\leq\Ucrk2(G)\leq 2k$.
}
\label{fig:diagdouble}
\end{figure}
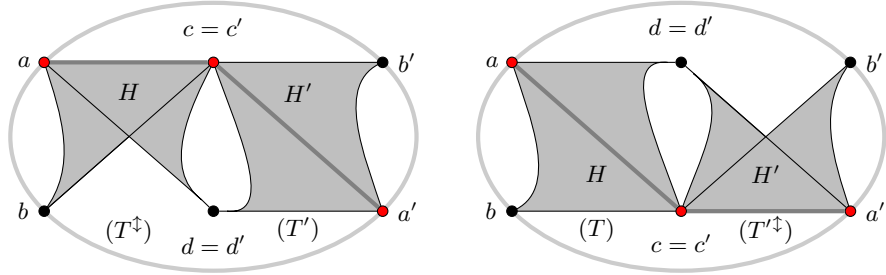

\begin{theorem}\label{thm:uncrohard2}
Given an integer $m$ and a graph $G$, the problems to decide whether $\Ucr(G)\leq m$ and whether $\Ucrk2(G)\leq m$ are both \NP-complete.
This holds even if $G$ contains an edge $f\in E(G)$ such that $G-f$ is planar.
\end{theorem}

\begin{proof}
Let $T=(H,a,b,c,d)$ be a tile as expected in \Cref{thm:diaghard1}, from which we reduce, and let $T'=(H',a',b',c',d')$ be a copy of $T$.
We take a $4$-cycle $C$ on the vertices $(a,b,a',b')$ in this order, such that each edge of $C$ is of weight $2k+1$, 
or equivalently composed of $2k+1$ parallel edges, and make $G$ as the union $C\cup H\cup H'$ with the identification $c=c'$ and $d=d'$.
We may pick $f$ as one of the two edges incident to $d$ and see that $G-f$ is planar.
Let~$m=2k$.

If $\tcrn(T{}^{\updownarrow})\leq k$, then we have a drawing $D_1$ of $G$ such that; $C$ is uncrossed,
$H$ is drawn as in the inverted tile $T{}^{\updownarrow}$ with $\leq k$ crossings and $H'$ is drawn as the planar tile $T'$ with no crossings.
See \Cref{fig:diagdouble}.
We let $D_2$ be the drawing symmetric to $D_1$, with $H'$ drawn as in the inverted tile $T'{}^{\updownarrow}$ and $H$ without crossings.
The collection $\{D_1,D_2\}$ certifies that \[\Ucr(G)\leq\Ucrk2(G)\leq k+k=m.\]

On the other hand, assume an uncrossed collection of at least two drawings of $G$ with at most $m$ crossings in total.
Then one of these drawings, say $D_1$, has at most $m/2=k$ crossings.
In particular, the cycle $C$ of weight $2k+1$ must be uncrossed (as a crossing ``hits'' at most two of the parallel edges composing each edge of~$C$),
and so $C$ makes the boundary of a valid tile drawing of connected $(H\cup H',a,b,a',b')$. 
Furthermore, by the assumption of \Cref{thm:diaghard1} on paths from $a$ to $c$, there exist $2k+1$ 
edge-disjoint paths from $a$ to $a'$ in $H\cup H'$, and again there is an $a$--$a'$ path $P$ without crossings on it in~$D_1$.

Consider now the position of the vertex $d=d'$ with respect to the uncrossed path $P$ in $D_1$; up to symmetry between $H$ and $H'$, let $d=d'$ lie ``below'' $P$
(with respect to having $C$ depicted with the vertices $a$ and $b'$ at the top, and $b,a'$ at the bottom, as in \Cref{fig:diagdouble} left).
Let $P'=P\cap H'$.
Let $Q'\subseteq H'$ be a shortest path from $d=d'$ to $V(P')$, which exists in a perfectly connected tile (we do not require $Q'$ to be crossing-free).
We now modify the drawing $D_1$ into $D_1'$ as follows; we delete the subdrawing of $H'$ except $P'\cup Q'$ from $D_1$,
then we contract the crossing-free path $P'$ making $c=a'$, and add a new vertex $d''$ in a close neighbourhood of $a'$ such that
$d''$ is adjacent only to $d$ by an edge drawn along $Q'$ (and then $Q'$ is deleted).
This does not introduce any more crossings than existed in $D_1$.
Since $d$ was of degree $1$ in $H$ and is of degree $2$ in $D_1'$, we may now simply identify $d=d''$ by prolonging the single edge of $d$ in~$H$.
By the assumption of $d$ being ``below'' $P$, we get that the cyclic order of the tile vertices in $D_1'$ is now $(a,b,d,c)$, 
and so $D_1'$ is a valid tile drawing of $(H,a,b,d,c)=T{}^{\updownarrow}$ with at most $\crd(D_1')\leq\crd(D_1)\leq k$ crossings.

Summarizing the previous, $\Ucr(G)\leq m$ or $\Ucrk2(G)\leq m$ implies $\tcrn(T{}^{\updownarrow})\leq\crd(D_1)\leq k$.
On the other hand, membership in \NP{} is shown in the standard way; we guess an uncrossed collection of drawings,
imagine every crossing point turned into a new vertex subdividing the participating edges, and test planarity.
  This completes the proof.
\end{proof}

\section{\FPT{} algorithm}\label{sec:FPT}

We now give the second main result, that computing the uncrossed crossing number is fixed-parameter tractable in the solution value.

\begin{theorem}\label{thm:fptunccr}
There is a quadratic \FPT{}-time algorithm with an integer parameter $k$ that, given a graph $G$ and an integer~$c$,
decides whether $\Ucrk c(G)\leq k$.
\end{theorem}
We first remark that only values of $c\leq k$ in \Cref{thm:fptunccr} are meaningful since, having an uncrossed
collection of $c>k$ drawings with the total number of crossings $k$, implies that one of the drawings is planar,
and so is $G$, which can be tested beforehand.
In particular, $\Ucr(G)\leq k$ $\iff$ $\Ucrk k(G)\leq k$.
Furthermore, if $\Ucr(G)\leq k$ in \Cref{thm:fptunccr}, then we can straightforwardly compute $\Um(G)$.

\Cref{thm:fptunccr} is analogous to the \FPT{} algorithm of Grohe~\cite{Grohe04} for the classical crossing number,
and our proof is on a high level very similar to~\cite{Grohe04}.
However, the technical details are different, not only because we need a slightly different formulation of the problem itself,
but also since the grid-reduction phase of the algorithm~\cite{Grohe04} uses annotation with uncrossable edges which
is not suitable when working with a collection of many drawings.

Let a {\em hexagonal graph} be a simple plane graph $H$ such that all faces of~$H$ are of length $6$, all vertices
of degree $2$ or $3$, and the vertices of degree $2$ occur only on the outer face of~$H$.
A hexagonal graph $H$ is called a {\em hexagonal $r$-grid} (also a ``wall'') if the following holds:
there are pairwise-disjoint cycles $C_1,\ldots,C_r\subseteq H$, called the {\em principal cycles},
such that $V(H)=V(C_1)\cup\cdots\cup V(C_r)$, the cycle $C_1$ bounds a hexagonal face,
and for $i=2,\ldots,r$ the cycle $C_i$ contains $C_{i-1}$ in its interior (they are nested from outermost $C_r$ to innermost~$C_1$).
See \Cref{fig:hexgrid}, and note that such $H$ is unique up to isomorphism for each $r$.

More generally, given a graph $G$, we call the hexagonal $r$-grid~$H$ a {\em hexagonal $r$-grid in $G$}, 
if there is a subgraph $H_1\subseteq G$ (also called a \emph{grid in~$G$}) which is isomorphic to a subdivision of $H$.
We canonically extend the meaning of principal cycles to~$H_1$.
The {interior vertices of $H$} are those not on the outermost principal cycle, and the \emph{interior vertices of~$H_1$} are those 
which are the interior vertices of $H$ or subdivisions of edges of $H$ incident to interior vertices of $H$.
A grid $H_1\subseteq G$ is \emph{flat} if the subgraph of $G$ induced on $V(H_1)$ and the vertices of all connected components
of $G-V(H_1)$ adjacent to (some) interior vertices of $H_1$ is planar.

\begin{figure}[tb]
	\centering
	\begin{tikzpicture}[scale=0.7]\small
	\tikzstyle{every node}=[draw, shape=circle, inner sep=1pt, fill=black]
	\foreach \x in {0,2.8,5.6,8.4} \foreach \y in {0,1.4,2.8,4.2}
		\draw (\x,\y) node {} -- (\x+1,\y) node {} -- (\x+1.4,\y+0.7) node {} 
		-- (\x+1,\y+1.4) node {} -- (\x,\y+1.4) node {} -- (\x-0.4,\y+0.7) node {} -- cycle;
	\foreach \x in {1.4,4.2,7} \foreach \y in {-0.7,2.1,4.9}
		\draw (\x,\y) node {} -- (\x+1,\y) node {} -- (\x+1.4,\y+0.7) node {} 
		-- (\x+1,\y+1.4) node {} -- (\x,\y+1.4) node {} -- (\x-0.4,\y+0.7) node {} -- cycle;
	\draw[very thick, dashed] (4.2,2.1) -- node[draw=none,fill=none,label=above:$\!C_1\!$] {} (5.2,2.1)
		 -- (5.6,2.8) -- (5.2,3.5) -- (4.2,3.5) -- (3.8,2.8) -- cycle;
	\draw[very thick, dashed] (4.2,0.7) -- node[draw=none,fill=none,label=above:$\!\!C_2\!\!$] {} (5.2,0.7)
		 -- (5.6,1.4) -- (6.6,1.4) -- (7,2.1) -- (6.6,2.8)
		 -- (7,3.5) -- (6.6,4.2) -- (5.6,4.2) -- (5.2,4.9) -- (4.2,4.9) -- (3.8,4.2) -- (2.8,4.2)
		 -- (2.4,3.5) -- (2.8,2.8) -- (2.4,2.1) -- (2.8,1.4) -- (3.8,1.4) -- cycle;
	\draw[very thick, dashed] (4.2,-0.7) -- node[draw=none,fill=none,label=above:$\!\!C_3\!\!$] {} (5.2,-0.7)
		 -- (5.6,0) -- (6.6,0) -- (7,0.7) -- (8,0.7)
		 -- (8.4,1.4) -- (8,2.1) -- (8.4,2.8) -- (8,3.5) -- (8.4,4.2) -- (8,4.9)
		 -- (7,4.9) -- (6.6,5.6) -- (5.6,5.6) -- (5.2,6.3) -- (4.2,6.3) -- (3.8,5.6) -- (2.8,5.6)
		 -- (2.4,4.9) -- (1.4,4.9) -- (1,4.2) -- (1.4,3.5) -- (1,2.8) -- (1.4,2.1) -- (1,1.4)
		 -- (1.4,0.7) -- (2.4,0.7) -- (2.8,0) -- (3.8,0) -- cycle;
	\end{tikzpicture}
	\caption{A fragment of a hexagonal graph, and the principal cycles $C_1,C_2,C_3$ (thick-dashed) defining a hexagonal $3$-grid.}
	\label{fig:hexgrid}
\end{figure}
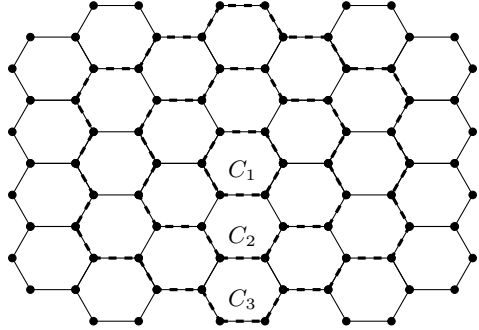

We proceed with the following steps.

\sv{%
\begin{theorem2rep}[\cite{ChuzhoyT21,PerkovicR00,RobertsonS95}]\apxmark\label{thm:gethexg}
For any $r\in\mathbb{N}$ there is a linear \FPT{}-time algorithm parameterized by $r$,
which given a graph $G$ outputs either a tree decomposition of $G$ of width at most $4^{10}r^{10}$, or a hexagonal $r$-grid in~$G$.
\end{theorem2rep}
}
\lv{%
\begin{theorem}[\cite{ChuzhoyT21,PerkovicR00,RobertsonS95}]%
  \label{thm:gethexg}
For any $r\in\mathbb{N}$ there is a linear \FPT{}-time algorithm parameterized by $r$,
which given a graph $G$ outputs either a tree decomposition of $G$ of width at most $4^{10}r^{10}$, or a hexagonal $r$-grid in~$G$.
\end{theorem}
}
	The statement is a combination of the following:
	\begin{theorem}[Chuzhoy and Tan~\cite{ChuzhoyT21}]
		For any $r \in \mathbb{N}$, there is an $r$-grid in any graph of treewidth at least $4^{10}r^{10}$.
	\end{theorem}
	
	\begin{theorem}[Robertson and Seymour~\cite{RobertsonS95}, Perkovi\'c and Reed~\cite{PerkovicR00}]
		Let $g$ be a computable function such that, for each $r\in \mathbb{N}$, there is an $r$-grid in any graph of treewidth at least $g(r)$.
		For any $r\in \mathbb{N}$, there is a linear-time algorithm which, given a graph $G$, outputs either a tree decomposition of $G$ of width at most $g(r)$, or an $r$-grid in~$G$.
	\end{theorem}

\sv{\begin{proof}[outline]}
  \lv{\begin{proof}[Outline of proof of \Cref{thm:gethexg}]}
	The preceding theorem was originally proved with quadra\-tic running ti\-me~\cite{RobertsonS95}, which was required for finding an $r$-grid in the case of treewidth at least $g(r)$.
	However, this can be avoided by using a linear-time algorithm to find a subgraph $H$ of $G$ with treewidth at least $g(r)$ and at most $2g(r)$ whenever the treewidth is at least~$g(r)$.
	Then a $r$-grid exists in $H$ by the precondition of the theorem and can be found in linear \FPT{}-time parameterized by $r$ in $H$ by using Courcelle's theorem~as~in~\cite{KneisL09}.
	In a combination, we thus get Theorem~\ref{thm:gethexg}.
\end{proof}

\begin{lemma}[\cite{DBLP:journals/jct/Thomassen97b,Grohe04}]\label{lem:flatgrid}
For $r_1,k\in\mathbb{N}$ and a graph $G$ we can in linear time compute the following:
Given a hexagonal $r$-grid in $G$, where $r=400r_1\lceil\sqrt{k}\rceil$, either output a flat hexagonal $r_1$-grid in $G$,
or answer truthfully that~\mbox{$\crg(G)\geq k+1$}.
\end{lemma}

\sv{%
\begin{lemma2rep}[\cite{DBLP:journals/ejc/GeelenRS04}]\label{lem:gridback}\apxmark
For $k\in\mathbb{N}$, let a graph $G$ contain a flat hexagonal $(4k+4)$-grid $H\subseteq G$, and let $e\in E(H)$
be an edge of the innermost principal cycle of~$H$.
If there is a drawing $D_1$ of $G-e$ such that $\crd(D_1)\leq k$, then there is a drawing $D$ of $G$ such that
every crossing of $D$ exists also in $D_1$. (In particular, $\crd(D)\leq\crd(D_1)$ and $e$ is not crossed in~$D$.)
\end{lemma2rep}
}
\lv{%
\begin{lemma}[\cite{DBLP:journals/ejc/GeelenRS04}]\label{lem:gridback}%
For $k\in\mathbb{N}$, let a graph $G$ contain a flat hexagonal $(4k+4)$-grid $H\subseteq G$, and let $e\in E(H)$
be an edge of the innermost principal cycle of~$H$.
If there is a drawing $D_1$ of $G-e$ such that $\crd(D_1)\leq k$, then there is a drawing $D$ of $G$ such that
every crossing of $D$ exists also in $D_1$. (In particular, $\crd(D)\leq\crd(D_1)$ and $e$ is not crossed in~$D$.)
\end{lemma}
}
\sv{\begin{proof}[outline]}
  \lv{\begin{proof}[Outline of proof]}
The statement is only implicit in \cite[Section~3]{DBLP:journals/ejc/GeelenRS04}, however, with \Cref{lem:flatgrid}
at hand, the application of the ideas of \cite{DBLP:journals/ejc/GeelenRS04} is very straightforward.

Recall that, in the grid $H\subseteq G$, we have the principal cycles $C_1,C_2,\ldots,C_{4k+4}$ ordered from innermost $C_1$ to outermost $C_{4k+4}$.
We call a \emph{principal ring} $R_i$ the subgraph of $G$ induced by the following:
for an even index $i\leq 4k+2$, by the principal cycles $C_i\cup C_{i+1}$, and by all connected components of $G-V(H)$, which are adjacent to some interior vertices of the principal cycle $C_{i+2}$, but not adjacent to interior vertices of~$C_i$.
Then the cycle $C_1$ and the $2k+1$ rings $R_2,R_4,\ldots,R_{4k+2}$ are pairwise disjoint.

Since $\crd(D_1)\leq k$, one of the rings, say $R_i$, is uncrossed in~$D_1$, and so we may replace the apropriate part 
of the drawing $D_1$ with the planar subdrawing of $G$ in the flat grid $H$, precisely the planar subdrawing bounded by the ring~$R_i$.
(This way, we ``get'' the edge $e$ back into the drawing~$D$ from~$D_1$.)
\end{proof}

For the next lemma, we briefly mention that MSO$_2$ logic of graphs is logic dealing with multigraphs as two-sorted
relational structures with vertices, edges and their incidence relation.
In addition to standard propositional logic, MSO$_2$ quantifies over vertices and edges and their sets.
We write $H\models\varphi$ to mean that formula $\varphi$ evaluates to true on~$H$.
For a graph $H$, we denote by $H^{\bullet s}$ the graph obtained from $H$ by subdividing every edge with $s$ new vertices.
Closely following the ideas of Grohe~\cite{Grohe04}, we prove:
\sv{%
  \begin{lemma2rep}\label{lem:ucrmso}\apxmark
For $k,c\in\mathbb{N}$ such that $c\leq k$, there is a formula $\varphi_{c,k}$ in the MSO$_2$ logic,
such that for every graph $H$ we have $\Ucrk c(H)\leq k$ if and only if $H^{\bullet k}\models\varphi_{c,k}$.
\end{lemma2rep}
}
\lv{%
  \begin{lemma}\label{lem:ucrmso}%
For $k,c\in\mathbb{N}$ such that $c\leq k$, there is a formula $\varphi_{c,k}$ in the MSO$_2$ logic,
such that for every graph $H$ we have $\Ucrk c(H)\leq k$ if and only if $H^{\bullet k}\models\varphi_{c,k}$.
\end{lemma}
}

\begin{proof}
We closely follow the ideas of Grohe~\cite{Grohe04} when constructing $\varphi_{c,k}$.
We first remark that evaluating $\varphi_{c,k}$ on $H^{\bullet k}$ greatly simplifies the construction of $\varphi_{c,k}$,
but it is also possible to give a formula ``working'' directly on~$H$.

Since vertices of degree $2$ or $1$ do not play any role in evaluating crossings in our context, we may assume that
$H$ has minimum degree $3$, and so we can easily distinguish in $H^{\bullet k}$ between the original vertices of $H$
and the new \emph{subdivision vertices}.
Formula $\varphi_{c,k}$ starts with guessing, via existential quantifiers, $k$ vertex pairs $x_i,y_i$,~$i=1,\ldots,k$,
and verifying that these are only subdivision vertices, and that $x_1,\ldots,x_k$ are pairwise distinct and $y_1,\ldots,y_k$ are pairwise distinct.
Then the formula guesses vertex sets~$Z_j$, $j=1,\ldots,c$, and verifies that these sets partition the set $\{x_1,\ldots,x_k\}$.
Moreover, for every vertex set $S$ in $H^{\bullet k}$ such that $S$ is the set of subdivision vertices of some edge $e$ of~$H$
(which is easy to indentify as a maximal path of all internal degrees~$2$),
the formula verifies that for some $1\leq j\leq c$, the set $S$ is disjoint from $Z_j\not=\emptyset$ and from all $y_i$ such that~$x_i\in Z_j$
-- to mean that the edge $e$ is uncrossed in the corresponding drawing.

With respect to the previous, we denote by $H'_1,\ldots,H'_c$ the graphs, where $H'_j$ results from $H^{\bullet k}$ 
by identifying the vertex $x_i$ with~$y_i$ for $i\in\{1,\ldots,k\}$ such that~$x_i\in Z_j$
(note that no identification happens when $x_i=y_i$ in $H^{\bullet k}$).
By standard means, we can interpret each of the graphs $H'_j$ in $H^{\bullet k}$ within our formula,
and then test for planarity of each $H'_j$ such that $Z_j\not=\emptyset$ by looking for subdivisions of $K_5$ or~$K_{3,3}$
(see, e.g.,~\cite{Grohe04} for details of both tasks).
This finishes~$\varphi_{c,k}$.

Consider now that $H^{\bullet k}\models\varphi_{c,k}$.
By the above construction of the formula, this means that we have got a collection of $\leq\!c$ drawings of $H^{\bullet k}$,
one defined by each $H'_j$ such that $Z_j\not=\emptyset$, that every edge of $H$ is not participating in a crossing is some of these drawings,
and that there are altogether at most $k$ crossings in this collection.
Hence, $\Ucrk c(H)=\Ucrk c(H^{\bullet k})\leq k$.
On the other hand, having a collection of $\leq\!c$ drawings of $H$ witnessing $\Ucrk c(H)\leq k$, we can assign
the crossings in the collection to pairs of pairwise distinct vertices of $H^{\bullet k}$, which subdivide the
pairs of edges involved in the crossings and in the right order of crossings on each edge, such that each of
the graphs $H'_j$ resulting from identi\-fication will be planar, and hence this assignment will satisfy $\varphi_{c,k}$ on $H^{\bullet k}$.
\end{proof}

Now, we can finish the main task.
\sv{\begin{proof}[of \Cref{thm:fptunccr}]}
\lv{\begin{proof}[Proof of \Cref{thm:fptunccr}]}
Let $r\coloneqq400r_1\lceil\sqrt{k}\rceil$ where $r_1\coloneqq4k+4$, and $G_1\coloneqq G$.
Our algorithm first repeats the following until the listed treewidth condition on $G_1$ is met.
\begin{enumerate}[i.]%
\item Call the algorithm of Theorem~\ref{thm:gethexg} to compute a hexagonal $r$-grid $H_0\subseteq G_1$,
unless the tree-width of $G_1$ is at most $4^{10}r^{10}$.
Then apply \Cref{lem:flatgrid} onto $H_0$ to compute a flat hexagonal $r_1$-grid $H_1\subseteq G_1$,
unless $\crg(G_1)\geq k+1$.
\item Pick $e\in E(H_1)$ belonging to the innermost principal cycle of~$H_1$ arbitrarily and delete it,~$G_1'=G_1-e$.
Trivially, $\Ucrk c(G_1')\leq\Ucrk c(G_1)$, and from Lemma~\ref{lem:gridback} it follows that $\min(\Ucrk c(G_1'),k+1)=\min(\Ucrk c(G_1),k+1)$:
if $\Ucrk c(G_1')\leq k$, every drawing of the collection witnessing $\Ucrk c(G_1')$ has at most $k$ crossings, and so
Lemma~\ref{lem:gridback} is applicable and gives an (again) uncrossed collection of drawings of $G_1$, which implies $\Ucrk c(G_1)\leq\Ucrk c(G_1')$.
\item Continue to (i.) with $G_1\coloneqq G_1'$, i.e., removing the edge $e$.
\end{enumerate}

When the first routine is finished, we get $G_1\subseteq G$ such that;
\begin{itemize}%
\item $\crg(G_1)\geq k+1$, and hence $\crg(G)\geq k+1$ directly implying $\Ucrk c(G)\geq k+1$,
\item or the tree-width of $G_1$ is bounded by $4^{10}r^{10}$, and \[\min(\Ucrk c(G_1),k+1)=\min(\Ucrk c(G),k+1).\]
\end{itemize}
In the latter case, we apply classical Courcelle's theorem \cite{Courcelle90}
to decide whether $G_1^{\bullet k}\models\varphi_{c,k}$ for the formula of Lemma~\ref{lem:ucrmso}, using the tree decomposition
of $G_1$ computed by Theorem~\ref{thm:gethexg} and made into a tree decomposition of $G_1^{\bullet k}$.
This~is~the sought answer since, by previous, we have $\Ucrk c(G)\leq k$ $\iff$ $\Ucrk c(G_1)\leq k$.

Regarding runtime, we perform at most $|E(G)|$ iterations in the first stage, and each takes linear \FPT{}-time.
In the second stage we spend again linear \FPT{}-time, since testing $G_1^{\bullet k}\models\varphi_{c,k}$ is in linear \FPT{} with respect to the treewidth and the formula size as parameters.
\end{proof}

\begin{remark}
One can easily adapt Theorem~\ref{thm:gethexg} to actually compute a collection of drawings witnessing $\Ucrk c(G)\leq k$.
First, the constructive version of Courcelle's theorem, see~\cite{KneisL09}, run on $G_1\models\varphi_{c,k}$ computes witnesses of 
the positive answer in the form of a specification of which pairs of edges of $G_1$ cross and in which order in a collection giving $\Ucrk c(G_1)\leq k$.
By Lemma~\ref{lem:gridback}, only these crossings from $G_1$ are relevant for the corresponding uncrossed collection
of drawings of $G$, and hence we simply turn these crossings of $G_1$ into new vertices in $G$, and apply any standard planarity algorithm.
\end{remark}

\section{Conclusions}

Our work brings an interesting connection between two significantly studied concepts in the graph drawing area; the crossing number and the (outer)thick\-ness of graphs.
On the side closer to the classical crossing number, we suggest the uncrossed crossing number, and closer to the thickness, we propose the uncrossed number.
The initial results show that the uncrossed crossing number behaves quite similarly to the classical crossing number---and hence it seems
similarly hard to understand it in full generality. 
This may spark further interest in combinatorics research, and we hope that particular aspects of the uncrossed crossing number may be found useful in graph drawing.

In the combinatorial direction, for instance, it would be interesting to understand the structure of graphs critical for the uncrossed crossing number.
The results of \Cref{sec:FPT}, after an easy modification, also imply that these critical graphs do not contain large grids, and so have bounded tree-width.
Hence, one should ask whether critical graphs for the uncrossed crossing number resemble the already known asymptotic structure
of classical crossing-critical graphs~\cite{DBLP:conf/compgeom/DvorakHM18}.

\smallskip
On the other hand, our understanding of the uncrossed number is much more fuzzy and difficult to understand.
As we discussed, it can be 2-approxi\-ma\-ted by the outerthickness, but more detailed behavior is yet to be comprehended.
Therefore, based on our initial observation, we conjecture several hypotheses regarding this number.

We conjecture that the uncrossed number is \NP-hard (probably even para \NP-hard).
As the \NP-hardness of the outerthickness is surprisingly still unknown, exploring a related notion may give us more understanding to attack this problem as well.
Recently, there have been a progress in the follow-up paper~\cite{GDfollowup} showing that if the outherthickness problem is NP-hard so is the uncrossed number.
Probably even more challenging seems the understanding of the complexity of the crossing-optimal uncrossed number.

As the precise bound for thickness~\cite{Alekseev,BeinekeHararyMoon1964,Vasak} and outerthickness~\cite{Guy1990,GuyOuterthickness} of complete and complete bipartite graphs have been determined (almost\footnote{There are some possible exceptions in case of $\theta(K_{m,n})$ where the exact value is not known; see~\cite{Poranen} for further discussion.}) completely, we suggested exploring it for the uncrossed number.
Those numbers were completely determined in the follow-up paper~\cite{GDfollowup}.
Contrary to our original believes~\cite{OurGD}, there are values of $m\le n$ such that for complete bipartite graph $K_{m,n}$, we have
$\U(K_{m,n})\neq \theta_o(K_{m,n})$.
In particular, the uncrossed number is smaller by one for the infinite number of particular cases of $m$ and $n$ in the range $2m-2\le n$.
Consult~\cite{GDfollowup} for more details.
On the other hand, they confirmed that $\U(K_n)=\theta_o(K_n)$.
Such a behaviour brings up an interesting question whether a difference of those numbers can be arbitrary or not
in general graphs.

Another interesting direction is determining precise upper and lower bounds based on the minimum/maxi\-mum degree, as was done for outerthickness~\cite{Poranen} and thickness~\cite{Halton,Sykora,Wessel}.
A first non-trivial (better than thickness) lower-bound for general graphs was just given in~\cite[Theorem 3]{GDfollowup}:\[
  \U(G)\ge\left\lceil\frac{2m}{6n-10+\sqrt{(3n-5)^2-4m}}\right\rceil,
  \] where $n$ is the number of vertices and $m$ is the number of edges of $G$.
  This bound performs the better the more dense the graph is.
However, it is not clear how sharp the bound is and it would be very desirable to improve our understanding of this direction.

On the big picture, we proposed a general meta-problem that might highlight various aspects of a graph by providing a collection of drawing such that a particular property is met at least in one drawing within the collection.
In this paper, we proposed the first two(three) specific problems that capture this initial motivation, both naturally extending a seminal notion of the crossing number. 
We would like to see other possible approaches of what property one could impose on a collection of drawings that provide us with some new perspectives on the graph's structure.

\bibliographystyle{plainurl}
\bibliography{cro.bib}

\end{document}